\documentclass[10pt]{article}
\usepackage{amsmath}
\usepackage{amsthm}
\usepackage{amssymb}
\usepackage{enumerate}
\usepackage{epsf}
\usepackage{psfrag}
\usepackage{mathrsfs}
\DeclareGraphicsExtensions{.eps,.art,.ART,.ps}
\newcommand{\grad}{\operatorname{grad}}

\newtheorem{Thm}{Theorem}[section]

\begin{document}
\title{Test fields cannot destroy extremal black holes}
\author{Jos\'e Nat\'ario, Leonel Queimada and Rodrigo Vicente\\
{\small Center for Mathematical Analysis, Geometry and Dynamical Systems,} \\
{\small Mathematics Department, Instituto Superior T\'ecnico,} \\
{\small Universidade de Lisboa, Portugal}}
\date{}
\maketitle
\begin{abstract}
We prove that (possibly charged) test fields satisfying the null energy condition at the event horizon cannot overspin/overcharge extremal Kerr-Newman or Kerr-Newman-anti de Sitter black holes, that is, the weak cosmic censorship conjecture cannot be violated in the test field approximation. The argument relies on black hole thermodynamics (without assuming cosmic censorship), and does not depend on the precise nature of the fields. We also discuss generalizations of this result to other extremal black holes.
\end{abstract}
%
%
%
%
%
\section{Introduction}\label{section1}
In the wake of the proofs of the singularity theorems in general relativity \cite{P65, H67, HP70}, Penrose formulated the weak cosmic censorship conjecture \cite{P69, W97}, according to which, generically, the singularities resulting from gravitational collapse are hidden from the observers at infinity by a black hole event horizon. Penrose's expectation was that, independently of what might happen inside black holes, the evolution of the outside universe would proceed undisturbed.

To test this conjecture, Wald \cite{W74} devised a thought experiment to destroy extremal Kerr-Newman black holes, already on the verge of becoming naked singularities, by dropping charged and/or spinning test particles into the event horizon. Both him and subsequent authors \cite{TdFC76, Needham80} found that if the parameters of the infalling particle (energy, angular momentum, charge and/or spin) were suited to overspin/overcharge the black hole then the particle would not go in, in agreement with the cosmic censorship conjecture. Similar conclusions were reached by analyzing scalar and electromagnetic test fields propagating in extremal Kerr-Newman black hole backgrounds \cite{Semiz11, Toth12, DS13, Duztas14}. In this case, the fluxes of energy, angular momentum and charge across the event horizon were found to be always insufficient to overspin/overcharge the black hole. Some of these results have been extended to higher dimensions \cite{BCNR10} and also to the case when there is a negative cosmological constant \cite{GL15, RS14}.

More recently, it was noticed that Wald's thought experiment may produce naked singularities when applied to nearly extremal black holes \cite{Hubeny99, MS07, JS09, SS11}. However, in this case the perturbation cannot be assumed to be infinitesimal, and so backreaction effects have to be taken into account; when this is done, the validity of the cosmic censorship conjecture appears to be restored \cite{Hod08, BCK10, ZVPH13, SPAJ15, CBSM15}. It can also be argued that the third law of black hole thermodynamics \cite{BCH73}, for which there is some evidence \cite{Israel86, DN97, CLS10}, forbids subextremal black holes from ever becoming extremal, and so, presumably, from being destroyed. Nonetheless, this cannot be taken as a definitive argument, since, for instance, extremal Reissner-Nordstr\"om black holes can be formed by collapsing charged thin shells \cite{Boulware73}.

In this paper, we consider arbitrary (possibly charged) test fields propagating in extremal Kerr-Newman or Kerr-Newman-anti de Sitter (AdS) black hole backgrounds. Apart from ignoring their gravitational and electromagnetic backreaction, we make no further hypotheses on these fields: they could be any combination of scalar, vector or tensor fields, charged fluids, sigma models, elastic media, or other types of matter. This also includes test particles, since they can be seen as singular limits of continuous media \cite{GJ75, LGATN14}. We give a general proof that if the test fields satisfy the null energy condition at the event horizon then they cannot overspin/overcharge the  black hole. This is done by first establishing, in Section~\ref{section2}, a test field version of the second law of black hole thermodynamics for extremal Kerr-Newman or Kerr-Newman-AdS black holes (which does not assume cosmic censorship). We use this result in Section~\ref{section3}, together with the Smarr formula and the first law, to conclude the proof. This last step requires the black hole to be extremal, and cannot be extended to near-extremal black holes. In Section~\ref{section4} we discuss generalizations of our result to other extremal black holes, including higher dimensions and alternative theories of gravity.

We follow the conventions of \cite{MTW73, W84}; in particular, we use a system of units for which $c=G=1$.
\section{Second law for test fields}\label{section2}
In this section we prove that a version of the second law of black hole thermodynamics holds in the case of (possibly charged) test fields propagating on a background Kerr-Newman or Kerr-Newman-AdS black hole (either subextremal or extremal). This calculation is similar to the one in \cite{GW01}, but we do not assume cosmic censorship, i.e.\ we do not assume that the black hole is not destroyed by interacting with the test field.

We start by recalling the Kerr-Newman-AdS metric, given in Boyer-Lindquist coordinates by
\begin{align}
ds^2 = & - \frac{\Delta_r}{\rho^2}\left( dt - \frac{a \sin^2 \theta}{\Xi} d\varphi \right)^2 + \frac{\rho^2}{\Delta_r} dr^2 \nonumber \\
& + \frac{\rho^2}{\Delta_\theta} d \theta^2 + \frac{\Delta_\theta \sin^2\theta}{\rho^2}\left( a \, dt - \frac{r^2 + a^2}{\Xi} d\varphi \right)^2, \label{KNAdS}
\end{align}
where 
\begin{align}
& \rho^2 = r^2 + a^2 \cos^2 \theta; \\
& \Xi = 1 - \frac{a^2}{l^2}; \\
& \Delta_r = (r^2 + a^2)\left(1+\frac{r^2}{l^2}\right) - 2mr + q^2; \\
& \Delta_\theta = 1 - \frac{a^2}{l^2}\cos^2\theta
\end{align}
(see for instance \cite{CCK00}). Here $m$, $a$ and $q$ denote the mass, rotation and electric charge parameters, respectively. These parameters are related to the physical mass $M$, angular momentum $J$ and electric charge $Q$ by
\begin{equation}\label{physical}
M = \frac{m}{\Xi^2}, \qquad J = \frac{ma}{\Xi^2}, \qquad Q = \frac{q}{\Xi}.
\end{equation}
The cosmological constant is $\Lambda = - \frac{3}{l^2}$, and so the Kerr-Newman metric can be obtained by taking the limit $l \to +\infty$. To avoid repeating ourselves, we will present all calculations below for the Kerr-Newman-AdS metric only; the corresponding formulae for the Kerr-Newman metric can be easily retrieved by making $l \to +\infty$.

The Kerr-Newman-AdS metric, together with the electromagnetic $4$-potential
\begin{equation}
A = - \frac{qr}{\rho^2}\left(dt - \frac{a\sin^2\theta}{\Xi} d\varphi\right),
\end{equation}
is a solution of the Einstein-Maxwell equations with cosmological constant $\Lambda$. It admits a two-dimensional group of isometries, generated by the Killing vector fields $X = \frac{\partial}{\partial t}$ and $Y=\frac{\partial}{\partial \varphi}$. 

We consider arbitrary (possibly charged) test fields propagating in this background. Apart from ignoring their gravitational and electromagnetic backreaction, we make no further hypotheses on the fields: they could be any combination of scalar, vector or tensor fields, charged fluids, sigma models, elastic media, or other types of matter. Since the fields may be charged, their energy-momentum tensor $T$ satisfies the generalized Lorentz law\footnote{See the Appendix for a complete explanation of the origin and meaning of this equation.}
\begin{equation}\label{motion}
\nabla_\mu T^{\mu\nu} = F^{\nu\alpha} j_\alpha,
\end{equation}
where $F=dA$ is the Faraday tensor of the background electromagnetic field and $j$ is the charge current density $4$-vector associated to the test fields. Using the symmetry of $T$ and the Killing equation, 
\begin{equation}\label{Killing}
\nabla_\mu X_\nu + \nabla_\nu X_\mu = 0,
\end{equation}
we have
\begin{equation}
\nabla_\mu (T^{\mu\nu} X_\nu) = F^{\nu\alpha} j_\alpha X_\nu. \label{div1}
\end{equation}
On the other hand, using the charge conservation equation,
\begin{equation}\label{charge}
\nabla_\mu j^\mu = 0,
\end{equation}
we obtain
\begin{align}
\nabla_\mu (j^\mu A^\nu X_\nu) & =  j^\mu (\nabla_\mu A^\nu) X_\nu + j^\mu A^\nu \nabla_\mu X_\nu \nonumber \\
& = j^\mu (F_\mu^{\,\,\,\,\nu} + \nabla^\nu A_\mu) X_\nu - j^\mu A^\nu \nabla_\nu X_\mu \nonumber \\
& = F^{\mu\nu} j_\mu X_\nu + j_\mu (X^\nu \nabla_\nu A^\mu - A^\nu \nabla_\nu X^\mu). \label{div2}
\end{align}
Since $A$ is invariant under time translations, we have
\begin{equation}\label{Lie}
\mathcal{L}_X A = 0 \Leftrightarrow [X,A]= 0 \Leftrightarrow X^\nu \nabla_\nu A^\mu - A^\nu \nabla_\nu X^\mu = 0,
\end{equation}
and so from \eqref{div1} and \eqref{div2} we obtain
\begin{equation}
\nabla_\mu (T^{\mu\nu} X_\nu + j^\mu A^\nu X_\nu) = 0. \label{div3}
\end{equation}
This conservation law suggests that the total field energy on a given spacelike hypersurface $S$ extending from the black hole event horizon $\mathscr{H^+}$ to infinity (Figure~\ref{Penrose}) should be
\begin{equation}
E' = \int_S (T^{\mu\nu} + j^\mu A^\nu) X_\nu N_\mu dV_{3},
\end{equation}
where $N$ is the future-pointing unit normal to $S$. However, in the Kerr-Newman-AdS case the non-rotating observers at infinity are rotating with respect to the Killing vector field $X$ with angular velocity
\begin{equation}
\Omega_\infty = - \frac{a}{l^2},
\end{equation}
and so, as shown in \cite{Olea05}, the physical energy should be computed with respect to the non-rotating Killing vector field
\begin{equation}
K = X + \Omega_\infty Y = X - \frac{a}{l^2} Y,
\end{equation}
that is, the physical energy is actually
\begin{equation} \label{energy}
E = \int_S (T^{\mu\nu} + j^\mu A^\nu) K_\nu N_\mu dV_{3}.
\end{equation}
This correction was implemented for test particles in \cite{GL15}. The need for the corresponding correction in the calculation of the physical black hole mass has been stressed in \cite{GPP05, MO15}. Note that in the Kerr-Newman case $\Omega_\infty = 0$, and no correction is needed.

\begin{figure}[h!]
\begin{center}
\psfrag{i+}{$i^+$}
\psfrag{i0}{$i^0$}
\psfrag{i-}{$i^-$}
\psfrag{H}{$H$}
\psfrag{H+}{$\mathscr{H^+}$}
\psfrag{H-}{$\mathscr{H^-}$}
\psfrag{I+}{$\mathscr{I^+}$}
\psfrag{I-}{$\mathscr{I^-}$}
\psfrag{I}{$\mathscr{I}$}
\psfrag{S0}{$S_0$}
\psfrag{S1}{$S_1$}
\epsfxsize=0.9\textwidth
\leavevmode
\epsfbox{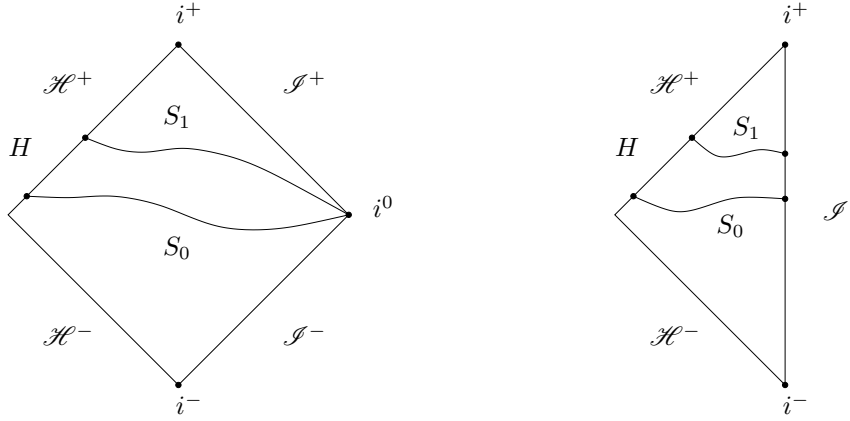}
\end{center}
\caption{Penrose diagrams for the region of outer communication of the Kerr-Newman (left) and Kerr-Newman-AdS (right) spacetimes.} \label{Penrose}
\end{figure}

Analogously, but now without ambiguity, the total field angular momentum on a spacelike hypersurface $S$ extending from the event horizon to infinity is
\begin{equation}
L = - \int_S (T^{\mu\nu} + j^\mu A^\nu) Y_\nu N_\mu dV_{3},
\end{equation}
where the minus sign accounts for the timelike unit normal. 

Consider now two such spacelike hypersurfaces, $S_0$ and $S_1$, with $S_1$ to the future of $S_0$ (Figure~\ref{Penrose}). We assume reflecting boundary conditions in the Kerr-Newman-AdS case, so that all fluxes vanish at infinity. The energy absorbed by the black hole across the subset $H$ of $\mathscr{H^+}$ between $S_0$ and $S_1$ is then
\begin{equation}
\Delta M = \int_{S_0} (T^{\mu\nu} + j^\mu A^\nu) K_\nu N_\mu dV_{3} - \int_{S_1} (T^{\mu\nu} + j^\mu A^\nu) K_\nu N_\mu dV_{3},
\end{equation}
whereas the angular momentum absorbed by the black hole across $H$ is
\begin{equation}
\Delta J = - \int_{S_0} (T^{\mu\nu} + j^\mu A^\nu) Y_\nu N_\mu dV_{3} + \int_{S_1} (T^{\mu\nu} + j^\mu A^\nu) Y_\nu N_\mu dV_{3}.
\end{equation}

Recall that the angular velocity of the black hole horizon is
\begin{equation}
\Omega_H = \frac{a \Xi}{r_+^2 +a^2},
\end{equation}
where $r_+$ is the largest root of $\Delta_r=0$. This means that the (future-pointing) Killing generator of $\mathscr{H^+}$ is
\begin{equation}
Z = X + \Omega_H Y = K + \Omega Y,
\end{equation}
where
\begin{equation}
\Omega = \Omega_H - \Omega_\infty
\end{equation}
is precisely the thermodynamic angular velocity, that is, the angular velocity that occurs in the first law for Kerr-Newman-AdS black holes \cite{GPP05}. Therefore, we have
\begin{equation}
\Delta M  - \Omega \Delta J = \int_{S_0} (T^{\mu\nu} + j^\mu A^\nu) Z_\nu N_\mu dV_{3} - \int_{S_1} (T^{\mu\nu} + j^\mu A^\nu) Z_\nu N_\mu dV_{3}.
\end{equation}
Because $Z$ is also a Killing vector field,
\begin{equation}
\nabla_\mu (T^{\mu\nu} Z_\nu + j^\mu A^\nu Z_\nu) = 0,
\end{equation}
and so the divergence theorem, applied to the region bounded by $S_0$, $S_1$ and $H$, yields
\begin{align} \label{Delta}
\Delta M  - \Omega \Delta J = \int_{H} (T^{\mu\nu} + j^\mu A^\nu) Z_\nu Z_\mu dV_{3}
\end{align}
(we use $-Z$ as the null normal\footnote{Recall that the divergence theorem on a Lorentzian manifold requires that the unit normal is outward-pointing when spacelike and inward-pointing when timelike. When the normal is null it is non-unique, and the volume element depends on the choice of normal; it should be past-pointing in the future null subset of the boundary, and future-pointing in the past null subset of the boundary.} on $H$).  Since on $\mathscr{H^+}$
\begin{equation}
A^\mu Z_\mu = - \frac{e r_+}{r_+^2 +a^2} = - \Phi,
\end{equation}
where $\Phi$ is the horizon's electric potential, we have
\begin{align}
\int_{H} j^\mu A^\nu Z_\nu Z_\mu dV_{3} = - \Phi \int_{H} j^\mu Z_\mu dV_{3}.
\end{align}
Using again the divergence theorem, this time together with the charge conservation equation \eqref{charge}, we obtain
\begin{align}
\int_{H} j^\mu A^\nu Z_\nu Z_\mu dV_{3} = -  \Phi \int_{S_0} j^\mu N_\mu  dV_{3} +  \Phi \int_{S_1} j^\mu N_\mu  dV_{3}.
\end{align}
Now the total charge on a spacelike hypersurface $S$ extending from the event horizon to infinity is
\begin{equation}
- \int_S j^\mu N_\mu dV_{3},
\end{equation}
where the minus sign accounts for the timelike unit normal. Therefore, denoting by $\Delta Q$ the electric charge absorbed by the black hole across $H$, we have
\begin{equation}
\int_{H} j^\mu A^\nu Z_\nu Z_\mu dV_{3} =  \Phi \Delta Q,
\end{equation}
and so equation \eqref{Delta} can then be written as
\begin{equation}
\Delta M  - \Omega \Delta J - \Phi \Delta Q  = \int_{H} (T^{\mu\nu} Z_\mu Z_\nu) dV_{3}.
\end{equation}

Since $Z$ is null on $H$, we have the following test field version of the second law of black hole thermodynamics. 

\begin{Thm}\label{Thm0}
If the energy-momentum tensor $T$ corresponding to any collection of test fields propagating on a Kerr-Newman or Kerr-Newman-AdS black hole background satisfies the null energy condition at the event horizon and appropriate boundary conditions at infinity then the energy $\Delta M$, angular momentum $\Delta J$ and electric charge $\Delta Q$ absorbed by the black hole satisfy
\begin{equation} \label{second}
\Delta M  \geq \Omega \Delta J + \Phi \Delta Q.
\end{equation}
\end{Thm}

It should be stressed that \eqref{second} is valid for extremal black holes, and it does not assume cosmic censorship, i.e.\ it does not assume that the Kerr-Newman-AdS metric with physical mass $M + \Delta M$, angular momentum $J + \Delta J$ and electric charge is $Q + \Delta Q$ represents a black hole rather than a naked singularity. Note that this scenario where the test fields interact with the geometry and change the values of the black hole charges is not in contradiction with the test field approximation, since the change is supposed to be infinitesimal.

%
%
%
\section{Proof of the result}\label{section3}
We can now prove our main result. 

\begin{Thm}\label{Thm1}
Test fields satisfying the null energy condition at the event horizon and appropriate boundary conditions at infinity cannot destroy extremal Kerr-Newman or Kerr-Newman-AdS black holes. More precisely, if an extremal black hole is characterized by the physical quantities $(M,J,Q)$, and absorbs energy, angular momentum and electric charge $(\Delta M,\Delta J,\Delta Q)$ by interacting with the test fields, then the metric corresponding to the physical quantities $(M+\Delta M, J + \Delta J, Q+\Delta Q)$ represents either a subextremal or an extremal black hole.
\end{Thm}

\begin{proof}
The physical mass of a Kerr-Newman or Kerr-Newman-AdS black hole, given in~\eqref{physical}, is completely determined by the black hole's event horizon area $A$, angular momentum $J$ and electric charge $Q$ through a Smarr formula
\begin{equation} \label{one}
M=M(A,J,Q).
\end{equation}
From the first law of black hole thermodynamics, we know that this function satisfies
\begin{equation} \label{two}
dM = \frac{\kappa}{8 \pi} dA + \Omega dJ + \Phi dQ,
\end{equation}
where $\kappa$ is the surface gravity of the event horizon \cite{BCH73, CCK00, GPP05}. The condition for the black hole to be extremal is
\begin{equation}
\kappa = 0 \Leftrightarrow \frac{\partial M}{\partial A}(A,J,Q)=0,
\end{equation}
which can be solved to yield the area of an extremal black hole as a function of its angular momentum and charge,
\begin{equation}
A=A_{\text{ext}}(J,Q).
\end{equation}
The mass of an extremal black hole with angular momentum $J$ and electric charge $Q$ is then
\begin{equation}
M_{\text{ext}}(J,Q) = M(A_{\text{ext}}(J,Q),J,Q).
\end{equation}
A Kerr-Newman-AdS metric characterized by $M$, $J$ and $Q$ will represent a black hole if $M \geq M_{\text{ext}}(J,Q)$, and a naked singularity if $M < M_{\text{ext}}(J,Q)$. We have
\begin{align}
dM_{\text{ext}} & = \left(\frac{\partial M}{\partial A} \frac{\partial A_{\text{ext}}}{\partial J} + \frac{\partial M}{\partial J}\right) dJ + \left(\frac{\partial M}{\partial A} \frac{\partial A_{\text{ext}}}{\partial Q} + \frac{\partial M}{\partial Q}\right) dQ \nonumber \\
& = \left(\frac{\kappa}{8 \pi} \frac{\partial A_{\text{ext}}}{\partial J} + \Omega \right) dJ + \left(\frac{\kappa}{8 \pi} \frac{\partial A_{\text{ext}}}{\partial Q} + \Phi\right) dQ \nonumber \\
& = \Omega dJ + \Phi dQ, \label{dMext}
\end{align}
where all quantities are evaluated at the extremal black hole. 

Consider now an extremal black hole with angular momentum $J$, electric charge $Q$ and mass $M = M_{\text{ext}}(J,Q)$. After interacting with the test fields, its angular momentum is $J + \Delta J$, its electric charge is $Q + \Delta Q$ and its mass is, using \eqref{second} and \eqref{dMext},
\begin{align}
M + \Delta M & \geq M + \Omega \Delta J + \Phi \Delta Q \nonumber \\
& = M_{\text{ext}}(J,Q) + \Delta M_{\text{ext}} \nonumber \\
& = M_{\text{ext}}(J + \Delta J,Q + \Delta Q).
\end{align}
In other words, the final mass is above the mass of an extremal black hole with the same angular momentum and electric charge, meaning that the final metric does not represent a naked singularity, that is, the black hole has not been destroyed.
\end{proof}
%
%
%
\section{Discussion}\label{section4}
In this paper we proved that extremal Kerr-Newman or Kerr-Newman-AdS black holes cannot be destroyed by interacting with (possibly charged) test fields satisfying the null energy condition at the event horizon and appropriate boundary conditions at infinity. This includes as particular cases all previous results of this kind obtained for scalar and electromagnetic test fields \cite{Semiz11, Toth12, DS13, Duztas14}. The corresponding results for test particles \cite{W74, TdFC76, Needham80} can also be considered particular cases, since particles can be seen as singular limits of continuous media \cite{GJ75, LGATN14}. It is interesting to note that if the null energy condition is not satisfied then the weak cosmic censorship conjecture may indeed be violated, as shown in~\cite{Duztas15, Toth15} for Dirac fields.

Our proof depends only on certain generic features of the Kerr-Newman or Kerr-Newman-AdS metric, and can therefore be adapted to other black holes. In fact, Theorem~\ref{Thm1} can be generalized as follows.

\begin{Thm}\label{Thm}
Consider a family of charged and spinning black holes in some metric theory of gravity, with suitable asymptotic regions, and test fields propagating in these backgrounds, such that:
\begin{enumerate}
\item
There exists an asymptotically timelike Killing vector field $K$, determining the black hole's physical mass, and angular Killing vector fields $Y_i$, yielding the black hole's angular momenta, such that event horizon's Killing generator is
\begin{equation}
Z = K + \sum_i \Omega_i Y_i,
\end{equation}
where $\Omega_i$ are the thermodynamic angular velocities (that is, the angular velocities that occur in the first law).
\item
There exists a Smarr formula relating the black hole's physical mass $M$, its entropy $S$, its angular momenta $J_i$ and its electric charge $Q$,
\begin{equation} \label{one}
M=M(S,J_i,Q),
\end{equation}
yielding the first law of black hole thermodynamics,
\begin{equation}
dM = T dS + \sum_i \Omega_i dJ_i + \Phi dQ,
\end{equation}
where $T$ is the black hole temperature and $\Phi$ is the event horizon's electric potential.
\item
Extremal black holes (that is, black holes with $T=0$) are characterized by $M=M_{\text{ext}}(J_i,Q)$, and subextremal black holes by $M>M_{\text{ext}}(J_i,Q)$.
\item
The test fields satisfy the null energy condition at the event horizon and appropriate boundary conditions at infinity.
\end{enumerate}
Then the test fields cannot destroy extremal black holes. More precisely, if an extremal black hole is characterized by the physical quantities $(M,J_i,Q)$, and absorbs energy, angular momenta and electric charge $(\Delta M,\Delta J_i,\Delta Q)$ by interacting with the test fields, then the metric corresponding to the physical quantities $(M+\Delta M, J_i + \Delta J_i, Q+\Delta Q)$ represents either a subextremal or an extremal black hole.
\end{Thm}

It is easy to check that this result applies to black holes in higher dimensions \cite{ER08}, including the case of a negative\footnote{Theorem~\ref{Thm} does not apply to the case of a positive cosmological constant, because the first hypothesis is not satisfied.} cosmological constant \cite{GPP05}. It can also be used for other black holes, like accelerated black holes with conical singularities \cite{Faraoni10} or black holes in alternative theories of gravity \cite{AGK16}. There is, however, no {\em a priori} reason why it should apply to arbitrary parametrized deformations of the Kerr metric \cite{CQ15}.
%
%
%
\section*{Acknowledgments}
JN was partially funded by FCT/Portugal through project PEst-OE/EEI/LA0009/2013. 
LQ gratefully acknowledges a scholarship from the Calouste Gulbenkian Foundation program {\em Novos Talentos em Matem\'atica}.
RV was supported by a graduate research fellowship from the FCT/Portugal project EXCL/MAT-GEO/0222/2012.
%
%
%
\section*{Appendix}
To obtain equation~\eqref{motion}, we observe that the charged test fields generate an extra electromagnetic field $f$ satisfying the Maxwell equations $df=0$ and
\begin{equation} \label{Maxwell}
\nabla_\mu f^{\mu\nu} = - j^\nu. 
\end{equation}
The total electromagnetic energy-momentum tensor is then
\begin{equation}
T^{EM}_{\mu\nu} = (F_{\mu\alpha}+f_{\mu\alpha})(F_{\nu}^{\,\,\,\,\alpha}+f_{\nu}^{\,\,\,\,\alpha}) - \frac14 g_{\mu\nu} (F_{\alpha\beta}+f_{\alpha\beta}) (F^{\alpha\beta}+f^{\alpha\beta}).
\end{equation}
Besides the stationary part, due solely to $F$, one has to consider, in the test field approximation, the cross terms
\begin{equation}\label{cross}
t_{\mu\nu} = f_{\mu\alpha} F_{\nu}^{\,\,\,\,\alpha} + F_{\mu\alpha} f_{\nu}^{\,\,\,\,\alpha} - \frac12 g_{\mu\nu} F_{\alpha\beta} f^{\alpha\beta}.
\end{equation}
We have
\begin{align}
\nabla^\mu t_{\mu\nu} & = - j_\alpha F_{\nu}^{\,\,\,\,\alpha} + f_{\mu\alpha} \nabla^\mu F_{\nu}^{\,\,\,\,\alpha} + F_{\mu\alpha} \nabla^\mu f_{\nu}^{\,\,\,\,\alpha} - \frac12 (\nabla_\nu F_{\alpha\beta}) f^{\alpha\beta} - \frac12 F_{\alpha\beta} \nabla_\nu  f^{\alpha\beta} \nonumber \\
& = - F_{\nu\alpha} j^{\alpha}, \label{tmunu}
\end{align}
where we used \eqref{Maxwell}, the Maxwell equation $\nabla^\mu F_{\mu\alpha} = 0$, the fact that
\begin{equation}
f_{\mu\alpha} \nabla^\mu F^{\nu\alpha} - \frac12 (\nabla^\nu F_{\alpha\beta}) f^{\alpha\beta} = \frac12 f_{\alpha\beta} \left( \nabla^\alpha F^{\nu\beta} + \nabla^\beta F^{\alpha\nu} - \nabla^\nu F^{\alpha\beta}\right) = 0
\end{equation}
(because of the Maxwell equation $dF=0$), and the fact that
\begin{equation}
F_{\mu\alpha} \nabla^\mu f^{\nu\alpha} - \frac12 F_{\alpha\beta} \nabla^\nu f^{\alpha\beta} = \frac12 F_{\alpha\beta} \left( \nabla^\alpha f^{\nu\beta} + \nabla^\beta f^{\alpha\nu} - \nabla^\nu f^{\alpha\beta}\right) = 0
\end{equation}
(because of the Maxwell equation $df=0$). Therefore, in the test field approximation, we have
\begin{equation}
\nabla^\mu \left(T_{\mu\nu} + T^{EM}_{\mu\nu} \right) = 0 \Leftrightarrow \nabla^\mu \left(T_{\mu\nu} + t_{\mu\nu} \right) = 0  \Leftrightarrow \nabla^\mu T_{\mu\nu} = F_{\nu\alpha} j^{\alpha},
\end{equation}
which is equation~\eqref{motion}.

One may wonder why not use the conserved current
\begin{equation}
\nabla_\mu (T^{\mu\nu} K_\nu + t^{\mu\nu} K_\nu) = 0
\end{equation}
to define the energy of the test field as
\begin{equation} \label{energy1}
E'' = \int_S (T^{\mu\nu} + t^{\mu\nu}) K_\nu N_\mu dV_{3}.
\end{equation}
The reason is that this expression accounts for the energy of the interaction between the charged field and the background electromagnetic field through the electromagnetic cross terms \eqref{cross}, whereas \eqref{energy} localizes it on the charges. As is well known, the physical mass of a charged black hole includes the energy of its background electromagnetic field; when charge enters the black hole, the interaction energy should be transferred from the energy of the electromagnetic field to the black hole's mass. This accounting is accomplished by \eqref{energy}, but not\footnote{As a toy model of this situation, consider a distribution of test charges with density $\rho$ on a background electrostatic field ${\bf E}=-\grad \phi$ generated by a closed surface kept at a constant potential $\Phi$. Using Gauss's law, it is easily seen that the total electrostatic energy outside the surface is $\int_{\text{out}} \rho \phi = \int_{\text{out}} {\bf E} \cdot {\bf e} - q_{\text{in}} \Phi$, where ${\bf e}$ is the electric field generated by the test charges and $q_{\text{in}}$ is the total test charge inside the surface.} by \eqref{energy1}.

One might also worry that the presence of the extra energy-momentum tensor $t$ with nonzero divergence \eqref{tmunu} could invalidate our previous conclusions. That is not the case, however, because $t$ does not contribute to the flux across the horizon. In fact, using \eqref{Lie} and the Killing equation \eqref{Killing}, we have
\begin{align}
& \int_H t_{\mu\nu} Z^\mu Z^\nu = \int_H  2 f_{\mu}^{\,\,\,\,\alpha} F_{\nu\alpha} Z^\mu Z^\nu = \int_H  2 f_{\mu}^{\,\,\,\,\alpha} (\nabla_\nu A_\alpha - \nabla_\alpha A_\nu) Z^\mu Z^\nu \nonumber \\
& = \int_H  2 f_{\mu}^{\,\,\,\,\alpha} (A^\nu \nabla_\nu Z_\alpha - Z^\nu \nabla_\alpha A_\nu) Z^\mu = \int_H 2 f_{\mu}^{\,\,\,\,\alpha} (- A^\nu \nabla_\alpha Z_\nu - Z^\nu \nabla_\alpha A_\nu) Z^\mu \nonumber \\
& = - \int_H  2 Z^\mu f_{\mu}^{\,\,\,\,\alpha} \nabla_\alpha (A^\nu Z_\nu) = 0,
\end{align}
since the vector field $Z^\mu f_{\mu}^{\,\,\,\,\alpha}$ is tangent to the event horizon,
\begin{equation}
Z^\mu f_{\mu}^{\,\,\,\,\alpha} Z_\alpha = 0,
\end{equation}
and $A^\nu Z_\nu = \Phi$ is constant along the event horizon.
%
%
%

\end{document}